\documentclass[]{imsart}

\RequirePackage{amsthm,amsmath,amsfonts,amssymb}
\RequirePackage[numbers]{natbib}

\startlocaldefs

\usepackage{graphicx}
\usepackage{fancyhdr}
\usepackage{enumerate}
\usepackage{amsmath}

\usepackage[multiple]{footmisc}

\usepackage{amsthm}
\usepackage{mathabx}
\usepackage{amssymb}
\usepackage{relsize}
\usepackage{pdfsync}
\usepackage[colorlinks=true,linkcolor=blue,citecolor=magenta]{hyperref}

\usepackage[normalem]{ulem}

\usepackage{xcolor}

\usepackage{cancel}

\usepackage{calligra}

\definecolor{darkgreen}{rgb}{0,0.5,0}

\usepackage{tikz}

\usepackage{changebar}

\usetikzlibrary{decorations.pathmorphing,patterns}
\usetikzlibrary{calc,arrows}

\usepackage{pgf}

\usepackage{multicol} 

\usetikzlibrary{automata}
\usetikzlibrary{positioning}

\tikzset{
    state/.style={
           rectangle,
           rounded corners,
           draw=black, very thick,
           minimum height=2em,
           inner sep=2pt,
           text centered,
           },
}

\usepackage{mathtools}
\usepackage{bbm}  

\newtheorem{theorem}{Theorem}[section]
\newtheorem{lemma}[theorem]{Lemma}

\newtheorem{remark}[theorem]{Remark}
\newtheorem{example}[theorem]{Example}

\newcommand\lang{\langle\hspace{-2pt}\langle}
\newcommand\rang{\rangle\hspace{-2pt}\rangle}

\newcommand\bbE{{\mathbb E}}
\newcommand\N{{\mathbb N}}

\newcommand\R{{\mathbb R}}
\newcommand\RR{{\mathbb R}}

\newcommand\EE{{\mathbb E}}

\newcommand\ve{\varepsilon}

\renewcommand{\ge}{\geqslant}
\renewcommand{\le}{\leqslant}

\renewcommand{\hat}{\widehat}
\newcommand{\tphi}{\widehat\varphi}
\newcommand{\ckphi}{\widecheck\varphi}
\renewcommand{\bar}{\overline}
\numberwithin{equation}{section}

\newcommand{\eqlaw}{\;\displaystyle{\mathop{=}^{\text{law}}}\;}
\newcommand{\tolaw}{\;\mathop{\longrightarrow}^{\text{law}}_{\ve\to0}\;}
\newcommand{\toas}{\;\mathop{\longrightarrow}^{\text{a.s.}}_{\ve\to0}\;}
\newcommand{\tonas}{\;\displaystyle{\mathop{\longrightarrow}^{\text{a.s.}}_{n\to\infty}}\;}

\newcommand{\tr}{{\tilde r}}

\newcommand{\tv}{{\tilde v}}

\newcommand{\tx}{{\tilde x}}
\newcommand{\ty}{{\tilde y}}

\DeclareMathOperator{\eff}{eff}

\allowdisplaybreaks
\mathtoolsset{showonlyrefs}
\endlocaldefs

\begin{document}

\begin{frontmatter}
\title{Macroscopic diffusive fluctuations for
  generalized hard rods dynamics}
\runtitle{Hard rods fluctuations}

\begin{aug}
  \author[A]{\fnms{Pablo}~\snm{A. Ferrari}\ead[label=e1]{pferrari@dm.uba.ar}}
  \and
\author[B]{\fnms{Stefano}~\snm{Olla}\ead[label=e2]{olla@ceremade.dauphine.fr}}

\address[A]{Pablo A. Ferrari, Departamento de Matemática,
  Facultad de Ciencias Exactas y Naturales,
  Universidad de Buenos Aires \\\emph{and}  IMAS-UBA-CONICET,
Buenos Aires, Argentina\printead[presep={,\ }]{e1}}

\address[B]{Stefano Olla, CEREMADE,
   Universit\'e Paris Dauphine - PSL Research University
\emph{and}  Institut Universitaire de France \\
\emph{and} GSSI, L'Aquila \orcid{0000-0003-0845-1861}\printead[presep={,\ }]{e2}}
\end{aug}

\begin{abstract}
We study the fluctuations in equilibrium for a dynamics of rods with random length.
  This includes the classical hard rod elastic collisions, when rod lengths are constant and equal to a positive value. We prove that in the diffusive space-time scaling,
  an initial fluctuation of density of particles of velocity $v$,
  after recentering on its Euler evolution, evolve randomly shifted by a Brownian motion
of variance $\mathcal D(v)$.
\end{abstract}



\end{frontmatter}


\section{Introduction}
\label{sec:introduction}

The mechanical system of one-dimensional hard rods is the simplest non trivial
completely integrable dynamics where the macroscopic behavior can be described
by generalized hydrodynamics. {A rod configuration $Y$ is a set of points $(y,v,r)$, viewed as a segment/rod $[y,y+r)$ travelling at bare speed $v$, satisfying $[y,y+r)\cap[\ty,ty+\tr)$ is empty for all $(y,v,r),(\ty,\tv,\tr)$ in $Y$. Rods travel ballistically with its bare velocity until collision with other rod, when they interchange speed and length}.
The density of particles of each given velocity
is conserved and in the Euler scaling the macroscopic evolution of such densities
have been studied by the pioneristic work of Percus \cite{P69} and
Boldrighini, Dobrushin and Suhov \cite{bds}.
Fluctuations around this Euler limit have been studied by Boldrighini and Wick in
\cite{boldrighini-wick}. Recently these results have been generalized to a completely
integrable dynamics of rods of random length (even negative length) where lengths
and velocities are exchanged at collision \cite{ffgs22}.
The elastic collisions are recovered in the particular case
that all rods have the same positive length.
In the case of rods with different or negative length this still defines a
 completely integrable dynamics, even though it does not correspond to a Hamiltonian
dynamics of particles.
Similar dynamics were considered in \cite{cardy-doyon},
while in \cite{ALM78} velocities are exchanged but not the lengths
(i.e. the classical elastic collision).

In this article we investigate the evolution of the densities fluctuations in
the diffusive space-time scale for the generalized dynamics considered in \cite{ffgs22}. We consider the system in a
\emph{stationary homogeneous}
initial condition. We will discuss initial inhomogeneous non-stationary state
in~\autoref{sec:inhom-non-stat}.

The result we prove in the present article
is that the initial fluctuations of the density of particles of velocity $v$,
after recentering on its Euler evolution, evolve randomly shifted by a Brownian motion
of variance $\mathcal D(v)$. This diffusion coefficient $\mathcal D(v)$
has an explicit expression depending on $v$ and on the particular stationary measure
(cf. \eqref{eq:83}).
In the case of rods of constant length $\mathcal D(v)$
is the same as computed by Spohn in
\cite{spohn1982hydrodynamical}, as well as it appears in the first order diffusive
correction to the Euler Hydrodynamic limit \cite{bds-1990-NS} \cite{MR1480034}.

This result corresponds to the following stochastic partial differential equation
for the evolution of the macroscopic fluctuation $\Xi_t(y,v,r)$
of the density of particles of velocity $v$ and length $r$ at position $y$:
\begin{equation}
  \label{eq:17}
  \partial_t \Xi_t(y,v,r) = \frac 12 \mathcal D(v) \partial_y^2 \Xi_t(y,v,r)
  + \partial_y \Xi_t(y,v,r)  \sqrt{\mathcal D(v)} \dot W_t(v), 
\end{equation}
where $(\dot W_t(v): t,v \in\R)$ is a centered Gaussian field with covariance
\begin{equation}
  \label{eq:18}
  \bbE\bigl(\dot W_t(v) \dot W_s(w)\bigr) = \delta(t-s)
  \frac{\Gamma(v,w)}{\sqrt{\mathcal D(v) \mathcal D(w)}}
\end{equation}
with $\Gamma(v,w)$ given in \eqref{eq:13}.
Since $\Xi_t$ is a distribution in
$(y,v,r)$, \eqref{eq:17} should be understood in the weak sense.

Notice that in \eqref{eq:17} the noise term is  a white noise in time,
in particular $W_t(v)$ does not depends on $y$.
This is in contrast with the typical
diffusive evolution of fluctuations in chaotic systems, where it is expected
an additive space-time white noise driving the fluctuations and the equation would be
of the type \cite{s}
\begin{equation}
  \label{eq:19}
   \partial_t \widebar \Xi_t(y,v,r) = \frac 12 \mathcal D(v) \partial_y^2 \bar \Xi_t(y,v,r)
  + \sqrt{\mathcal D(v)} \partial_y \dot W_{y,t}(v), 
\end{equation}
with
$\bbE\bigl(\dot W_{y,t}(v) \dot W_{y',s}(w)\bigr) = \delta(t-s)\delta(y-y')\delta(v-w)$.
Notice that the equilibrium solutions of \eqref{eq:17} and \eqref{eq:19}
  have the same space-time covariance $\mathbb E(\Xi_t(y,v,r) \Xi_s(y',v,r))$,
i.e. the space-time covariance does not give informations about
the martingale term of the evolution equation.
About the uncorrelation in $r$, this persists in the macroscopic fluctuations
as consequence of the decorrelation at the microscopic scale.

In order to understand why \eqref{eq:17} arise,
we follow the behaviour of two tagged quasi-particles with the same velocity.
We call here quasi-particles (or impulsions) the particles with the
dynamics defined by the exchange of positions at the moment of collision.
The standard technique to study such dynamics is to go to the
\emph{reduced description} where quasi-particles are mapped to
points and evolve without interaction.
Then the evolution of the tagged quasi-particle
is obtained by the trivial evolution of the
corresponding point, shifted by the collisions with quasi-particles
of different velocity.
Since the points corresponding to the other quasi-particles
are distributed by a Poisson field, these collisions happen
at random times and the collisional shifts are independent.
Consequently at the Euler scale we have a law of large numbers (see Section \ref{ref})
that produce a deterministic evolution of the tagged quasi-particle
with an effective velocity
given by \eqref{23} (such ergodicity was proven first in \cite{ALM75}).
Recentering the position of the tagged quasi-particle on the
Euler deterministic behavior, we have then a functional central limit theorem
so that the position converges in law to a Brownian motion of variance $\mathcal D(v)$.
Now consider two tagged quasi-particles with the same velocity and initially
located at macroscopic distance: this means that there are typically $\ve^{-1}$
particles in between, where $\ve$ is the scaling parameter going to $0$.
In the diffusive scaling each tagged quasi-particle at time $t$
has a number of collisions proportional
to $\ve^{-2}t$, but most of them with the same quasi-particles, except for an order
$\ve^{-1}$ of collisions at the beginning and at the end of the time interval $[0,t]$.
Consequently  the two tagged quasi-particles completely correlate
in the limit $\ve \to 0$, i.e. they converge to the same Brownian motion.
This causes the rigid motion of the
corresponding density fluctuations.
Notice that in chaotic system it is expected that
the two quasi-particles at initial macroscopic distance converge to two independent
Brownian motion, generating the space-time white noise present in \eqref{eq:19}.

As far as we know, equation \eqref{eq:17} for the diffusive fluctuations for hard rods
is new in the literature. 
We recently discovered the article by Presutti and Wick \cite{MR971036}
that concerns the diffusive behavior of travelling wave initial conditions in hard rods
systems,
where in the remark after Theorem 1 they comment about possible
diffusive behaviour of fluctuations and they write:
``spatially separated fluctuations in
the density of rods with the same velocity
move with the same Brownian component''.
Strangely \cite{MR971036} is never quoted in the following literature
about Navier-Stokes corrections for the hard rods hydrodynamics
(cf. \cite{MR1480034}). In the introduction of \cite{boldrighini-wick} it was
announced a second article about the \emph{Navier Stokes corrections}
for the evolution of the fluctuations, but the authors confirmed us that
this has never been written.

We believe 
that \eqref{eq:17} is a typical (universal) macroscopic
behaviour for the diffusive fluctuations of completely integrable many-body
systems \cite{spohn23}. 
Since this macroscopic behaviour is due to the scattering shift that the quasi particles have
  when \emph{crossing} particles with a different velocity, we believe that this equations are
  universal for any extended integrable system. In fact any extended integrable system is
  equivalent (in a macroscopic limit) to a collisional system
  (of \emph{quasi-particle} or \emph{impulsions} or \emph{solitons})
  and it is characterized macroscopically
  by its scattering shift \cite{spohn23}. 
  What makes the hard rods system treatable is the fact that the scattering shift
  does not depend on the velocities of the quasi-particles colliding.

Hard rods dynamics with
domain wall initial conditions, a particular case of travelling wave,
is investigated in  \cite{ds}; the paper includes the Navier-Stokes corrections and the computation of the covariance
$\Gamma(v,v')$ from Green-Kubo formula.

Diffusive corrections to the hydrodynamic Euler scaling in general completely integrable systems
  have been investigated in the physics literature, see the recent review \cite{DeNardis_2022}
  with the references therein, as well as \cite{huse-PRB.2018} and \cite{Doyon-SciPostPhys2019}.
  These articles contain general formulas for diffusion coefficients that are in agreement with our formula
  for the generalized hard rods. A generalization of our equation \eqref{eq:17} to other integrable systems
  will give answer about the macroscopic evolution on the diffusive scale
  of multipoint correlation functions, mentioned in \cite{DeNardis_2022} as an open problem.

As we believe that our approach is more elementary than the one used in the previous literature,
we have written this article to be completely independent of results on hard rods prior
to the paper \cite{ffgs22}, which is our starting point.  
Essentially, the only tools we use are the law of large numbers and the central limit theorem for a Poisson field.
In \autoref{sec:equil-stat-fluct} we prove the macroscopic evolution of
the fluctuations in the Euler scaling (recovering the result of \cite{boldrighini-wick}).
In \autoref{sec:equil-fluct-diff} we prove the evolution of the fluctuations in the
diffusive scaling. Finally in \autoref{sec:lemmaspp} %
we prove a lemma about limits for Poisson fields that we need in the proofs.

\section{Equilibrium Fluctuations in the Euler scaling}
\label{sec:equil-stat-fluct}

In the following $\mu$ is
a probability  on $\mathbb R^2$ with finite second moments, and  $\rho>0$ is a density parameter.
For $\ve>0$, let $X^\ve$ be the Poisson process on $\mathbb R^3$ with intensity $\ve^{-1} \rho dx\,d\mu(v,r)$.
We should think about $x$ as the \emph{macroscopic} position of the point $x$,
since the typical distance between points is $\ve$. The 
\emph{macroscopic length of the rod} $(x,v,r)$ is $\ve r$.

We define 
\begin{equation}
  \label{1}
  \begin{split}
    \sigma &:= \rho \iint r d\mu(v,r), \qquad \text{length density},\\
    \pi &:= \rho \iint r v d\mu(v,r), \qquad \text{momentum density}.
\end{split}
\end{equation}
We assume throughout the paper that the mean length density is positive, see \cite{ffgs22},
\begin{align}
  \label{s>0}
  \sigma>0.
\end{align}
The empirical length measure associated to $X^\ve$ is given by
\begin{align}\label{Nep}
  N^\ve\varphi := \sum_{(x,v,r)\in X^{\ve}} \ve r\varphi(x, v, r),\qquad N^\ve(A):= N^\ve1_A.
\end{align}
Here and in the following $\varphi(x,v,r)$ is a generic continuous
  function on $\mathbb R^3$,
smooth and compactly supported in $x$ and square integrable in $(v,r)$
  with respect to the measure $r^2 d\mu(v,r)$.
The expectation of $N^\ve\varphi$ is
\begin{align}
   \label{mu1}
  \EE N^\ve\varphi = \iiint \varphi(x,v,r)\; r\, \rho\,dx\, d\mu(v,r)=: \lang \varphi \rang,
\end{align}
where  $\lang\cdot\rang$ is called the length biased measure.

We have the classical law of large numbers for Poisson processes 
\begin{equation}
  \label{xlln}
 N^\ve\varphi   \toas 
\lang \varphi \rang.
\end{equation}
We sketch the proof of \eqref{xlln}. Let $(X_i)_{i\ge1}$ be a sequence of
  iid Poisson processes on $\R^3$ with intensity measure~$\rho\,dx\,d\mu(v,r)$,
  and empirical measures $N_i$ defined by $N_i\varphi:=\sum_{(x,v,r)\in X_i}r\,\varphi(x,v,r)$.
  Denoting $\ve_n=\frac1n$, with $n\in\N$, we can write $X^{\ve_n} = \dot\cup_{i=1}^n X_i$, by the superposition theorem \cite{kingman}. Then, $N^{\ve_n}\varphi  = \frac1n \sum_{i=1}^n N_i \varphi\displaystyle{\tonas} \lang \varphi \rang$,
  by the law of large numbers for iid random variables.

Considering $\varphi(x, v, r) = 1_{[0,1]}(x)$, \eqref{xlln} implies that $\ve$ times the number of points in $X^\ve$ with the $x$ coordinate in $[0,1]$ converges to the density $\sigma$, 
  $N^\ve([0,1]\times \RR^2) 
  \displaystyle{\toas} \sigma$.

  The convergence \eqref{xlln} can be extended to any function in $L^2(\lang\cdot\rang_2)$, the space of square integrable functions with respect to the measure  defined by
  \begin{align}
  \label{mu2}
 \lang \varphi \rang_2:=  \iiint  \varphi(x,v,r)\, r^2 \, \rho\,dx\, d\mu(v,r).
  \end{align}

  The fluctuation field is a random signed measure on $\R^3$ defined by
  \begin{equation}
    \label{eq:10}
    \xi^{X,\ve}(\varphi) :=
  \ve^{-1/2} \bigl(N^\ve\varphi - \lang \varphi \rang\bigr).
\end{equation}
For any $\ve>0$ and $\varphi,\psi \in L^2(\lang\cdot\rang_2)$ we can compute the covariance
\begin{equation}
  \label{eq:32}
  \mathbb E\bigl(\xi^{X,\ve}(\varphi) \xi^{X,\ve} (\psi)\bigr) =
  \lang \varphi \psi \rang_2 - \lang \varphi \rang \lang \psi \rang.
\end{equation}

Using the superposition theorem as for \eqref{xlln}, we have the following central limit theorem
\begin{equation}
  \label{xclts}
  \xi^{X,\ve}(\varphi) 
  \ \tolaw\ \xi^X(\varphi),
\end{equation}
where $\xi^X(\varphi)$ is a centered Gaussian random variable with variance
\begin{equation}
  \label{eq:28}
  \mathbb E\bigl(\xi^X(\varphi)^2 \bigr) = \lang \varphi^2 \rang_2 - \lang \varphi \rang^2,
\end{equation}
and, by \eqref{eq:32},
\begin{align}
  \label{4}
  \mathbb E\bigl(\xi^X(\varphi) \xi^X(\psi)\bigr) = \lang \varphi \psi \rang_2 - \lang \varphi \rang \lang \psi \rang.
\end{align}
By the multidimensional central limit theorem
  $\xi^X(\varphi)$ and $\xi^X(\psi)$ are jointly Gaussian.
This is enough to identify $\xi^X$ as a white noise on $\R^3$ with covariance \eqref{4},
i.e. as the centered random measure on  $\R^3$ such that, for any measurable $A,B\subset \R^3$
\begin{equation}
  \label{eq:34}
  \mathbb E\bigl(\xi^X(1_A) \xi^X(1_B)\bigr) =  \lang 1_{A\cap B} \rang_2.
\end{equation}

For any $a,b\in \R$, define the \emph{mass} (length) measure by 
\begin{equation}
  \label{5}
    m_a^b(X^\ve) =\ve\sum_{(x,v,r)\in X^\ve}  r\bigl(1_{[ x\in [a,b)]} - 1_{[ x\in [b,a)]}\bigr).
  \end{equation}
  Consequently,
  \begin{equation}
    \label{mab1}
     m_a^b(X^\ve)  \ \toas \   \EE m_a^b(X^\ve)=(b-a) \sigma.
\end{equation}
To each configuration $X^\ve$ and a point
$b\in \mathbb R$, there are a dilated point and configuration
\begin{align}
  D^\ve(b)&:= b + m_0^b(X^\ve) \\
Y^\ve =  D^\ve(X^\ve)&:=\{(D^\ve(x),v,r): (x,v,r)\in X^\ve\}.&& 
\end{align}
\begin{remark}
  \label{shift1}\rm
The distribution of $X^\ve$ is space shift invariant, but the distribution of the rod configuration $Y^\ve$ is not because $Y^\ve$ has no rod containing the origin. Our results can be extended to random rod configurations with space shift invariant distribution, by using Palm transforms \cite{thorisson} and Harris theorem \cite{Harris71}; see for instance \cite{fnrw21}.   
\end{remark}

The macroscopic dilation of the point $b$
is given by
\begin{equation}
  \label{2}
  \EE D^\ve(b)= b(1+\sigma).
\end{equation}

Denote  the length empirical measure induced by $Y^\ve$ by
\begin{align}
   K^\ve\varphi&:=\ve \sum_{(y,v,r)\in Y^\ve} r\varphi(y, v, r).
\end{align}
We have the law of large numbers:
\begin{align}
  K^\ve\varphi &=
    \ve \sum_{(x,v,r)\in X^\ve} r\varphi\bigl(x + m_0^x(X^\ve), v, r\bigr)\\
    & \toas\rho \iiint r\varphi\bigl(x + x \sigma,v,r\bigr) dx\; d\mu(v,r)\\
              &= \frac{\rho}{1 + \sigma}  \iiint r\varphi(y,v,r) dy\; d\mu(v,r)\\
  &=\frac{1}{1 + \sigma} \lang\varphi\rang.\label{115k}
  \end{align}

\subsection{Static CLT for the dilated configuration}
\label{sec:static-clt}

We define the fluctuation field
\begin{equation}
  \label{73}
  \xi^{Y,\ve}(\varphi) = \ve^{-1/2}\bigl(K^\ve\varphi - \EE K^\ve\varphi\bigr).
\end{equation}
We have
\begin{align}
  K^\ve\varphi - \EE K^\ve\varphi
  =  (K^\ve\varphi - A^\ve\varphi)
  + (A^\ve\varphi - \EE A^\ve\varphi)
  - (\EE K^\ve\varphi -\EE A^\ve\varphi). \label{117s}
\end{align}
where 
\begin{align}
 A^\ve\varphi&:= \ve \sum_{(x,v,r)\in X^\ve} r\varphi(x (1+\sigma), v, r),\\
\EE A^\ve\varphi&= \frac{\rho}{1 + \sigma} \iiint r\varphi(y, v, r) dy\,d\mu(v,r).
\end{align}
The last term in \eqref{117s} gives
\begin{align}
  &\vspace{-3mm} \ve^{-1/2} [\EE (K^\ve\varphi) -\EE (A^\ve\varphi)]\\
  &=  \mathbb E\Bigl( \ve^{1/2} \sum_{(x,v,r)\in X^\ve}\label{120s}
    r\bigl[\varphi\bigl( x +  m_0^x(X^\ve), v, r\bigr)
      - \varphi( x (1+\sigma), v, r)\bigr]\bigr)\\
  &=  \mathbb E\Bigl( \ve^{1/2} \sum_{(x,v,r)\in X^\ve}
    r(\partial_y \varphi)( x (1+\sigma), v, r)\bigl(m_0^x(X^\ve) - x\sigma\bigr)
     \Bigr) + R^\ve \label{890s}\\
&=  \iiint  r(\partial_y \varphi)( x (1+\sigma), v, r)
 \ve^{-1/2} \bigl[\mathbb E(m_0^x(X^\ve) - x\sigma)\bigr]  dx\; d\mu(v,r) + R^\ve \label{123s}\\
  &= R^\ve,\label{124s}
\end{align}
where $R^\ve$ denotes a generic term small with $\ve$.
Identity \eqref{123s} follows from Slyvniak-Mecke formula
(Theorem 3.2 in \cite{MR2004226}).
Since by \eqref{mab1} $\mathbb E(m_0^x(X^\ve) - x\sigma)=0$, 
 the identity \eqref{124s} follows.

 Defining
 \begin{equation}
   \label{eq:37}
   \ckphi (x,v,r) 
  := \varphi((1+ \sigma) x, v ,r),
 \end{equation}
 and applying \eqref{xclts} to the second term in \eqref{117s} we have
\begin{align}
  \notag
  &\vspace {-5mm}\ve^{-1/2}[A^\ve\varphi -\EE (A^\ve\varphi)]\\
  &= \ve^{-1/2}\Bigl[\ve \sum_{(x,v,r)\in X^\ve} r \ckphi ( x, v, r) -
    \rho \iiint r \ckphi (x,v,r) dx\; d\mu(v,r) \Bigr]\notag\\
   &\tolaw   \xi^X(\ckphi). \label{aea1s}
\end{align}
Notice that
\begin{align}
  \label{10s}
  &\hspace{-4mm}\EE (\xi^X(\ckphi) \xi^X(\widecheck\psi))\\
  &= \rho\iiint r^2 \ckphi (x,v,r)
    \widecheck\psi( x,v,r) dx d\mu(v,r) - \lang \ckphi\rang  \lang \widecheck\psi\rang\\
    &= \frac{\rho}{1+\sigma} \iiint r^2 \varphi(y,v,r) \psi(y,v,r) dy d\mu(v,r)
    - \frac 1{(1+\sigma)^2} \lang \varphi\rang  \lang \psi\rang\\
    &= \frac{1}{1+\sigma} \lang\varphi\psi\rang_2-\frac 1{(1+\sigma)^2}\lang \varphi\rang  \lang \psi\rang.
\end{align}

Finally we Taylor expand the first term of the RHS of \eqref{117s}:
  \begin{equation}
    \begin{split}
  \ve^{-1/2}(K^\ve\varphi - A^\ve\varphi)
    =   \ve^{1/2}\sum_{(x,v,r)\in X^\ve} r\bigl( \varphi\bigl(x + m_0^x(X^\ve), v, r\bigr)
   -\varphi(x (1+\sigma), v, r)\bigr)\\
   = \frac1{1+\sigma} \ve^{1/2}\sum_{(x,v,r)\in X^\ve} r (\partial_x \ckphi)( x, v, r)
   \bigl(  m_0^x(X^\ve) - x \sigma \bigr) + R^\ve
   \label{11s},
  \end{split}
\end{equation}
and the expectation of the rest $R^\ve$ is small again
  by using the Slyvniak-Mecke formula.
By the functional central limit theorem  \eqref{xclts} {we have for any $z\in \mathbb R$}
\begin{align}
  \label{13s}
  B^\ve(z) & := \ve^{-1/2}\bigl(  m_0^{z}(X^\ve) - z \sigma \bigr)
  =  \xi^{X,\ve}(1_{[0,z]}1_{z>0} - 1_{[z,0]}1_{z<0} ) \\
 & \tolaw
   \xi^{X}(1_{[0,z]}1_{z>0} - 1_{[z,0]}1_{z<0} ) =: B(z),
\end{align}
{where $(B(z):z\in \RR)$ is a bilateral Brownian motion. Using Lemma
\ref{LemmaPP} we have}
\begin{align}
  \label{14s}
  &\ve^{-1/2}(K^\ve\varphi - A^\ve\varphi)  \\
  &\tolaw \frac{\rho}{1+\sigma}\iiint r (\partial_x \ckphi)(x, v, r) B(x) dx \; d\mu(v,r)\label{131s}\\
    &= \frac 1{1+\sigma}
    \int dx B(x) \partial_x \Bigl(\rho \iint r \ckphi (x, v, r) \; d\mu(v,r)\Bigr)\\
    &=\frac 1{1+\sigma}  \int dx\; \xi^X\bigl(1_{[0,x]}1_{x>0} - 1_{[x,0]}1_{x<0}\bigr)
    \partial_x \Bigl(\rho \iint r \ckphi(x, v, r) \; d\mu(v,r)\Bigr)\\
    &=\frac 1{1+\sigma} \xi^X\Bigl( \int dx \bigl(1_{[0,x]}1_{x>0} - 1_{[x,0]}1_{x<0}\bigr)
    \partial_x \Bigl(\rho \iint r \ckphi (x, v, r) \; d\mu(v,r)\Bigr) \Bigr)\\
  &= - \frac {1}{1+\sigma}
  \xi^X\Bigl(\rho\, \iint r' \ckphi (\cdot, v', r') \; d\mu(v',r')\Bigr).\label{135s}
\end{align}
We conclude that
\begin{equation}
  \label{16s}
\ve^{1/2}(K^\ve\varphi - A^\ve\varphi) \tolaw  -  \xi^X(P \ckphi),
\end{equation}
where
\begin{equation}
  \label{15s}
  P\ckphi (x) = \frac{\rho}{1+\sigma} \iint r \ckphi(x, v', r') \; d\mu(v',r') =
   \frac{\rho}{1+\sigma} \iint r \varphi((1+\sigma) x, v', r') \; d\mu(v',r').
\end{equation}
Putting together \eqref{16s}, \eqref{120s}-\eqref{124s} and \eqref{aea1s}
we have shown that
\begin{equation}
  \label{17}
  \xi^{Y,\ve} (\varphi)\tolaw \xi^Y(\varphi) = \xi^X\Bigl(\ckphi -  P \ckphi\Bigr)
  = \xi^X\bigl(C \ckphi\bigr),
\end{equation}
where $C = I-  P$.
By the same argument used for $\xi^X$, this identifies $\xi^Y$ as the centered Gaussian field with covariance
\begin{align}
  \notag%
  \EE( \xi^Y(\varphi) \xi^Y(\psi)) \;&= \;\EE(\xi^X(C \widecheck\varphi) \xi^X(C \widecheck\psi))\\
   & =\rho  \iiint r^2 C \varphi(x(1+\sigma),v,r)
 C \psi(x(1+\sigma),v,r) dx d\mu(v,r),
    \\
&= \frac{\rho}{1+\sigma} \iiint r^2 C\varphi(y,v,r)
C\psi(y,v,r) dy d\mu(v,r)\\
&= \frac{\rho}{1+\sigma} \lang C\varphi C\psi \rang_2.
\end{align}
That means for the Fourier transforms 
\[\tphi(k,v,r) = \int e^{i2\pi k y} \varphi(y,v,r) dy
\]
\begin{equation}
  \label{18k}
  \begin{split}
    \EE(\xi^Y(\varphi) \xi^Y(\psi) )
    = \frac{\rho}{1+\sigma} \iiint r^2 C\hat\varphi(k,v,r)^*
  C\hat\psi(k,v,r) dk d\mu(v,r).
\end{split}
\end{equation}
i.e. a covariance operator
\begin{equation}
  \label{19}
 \mathcal C=  \frac{\rho}{1+\sigma} r^2C^2 =
  \frac{\rho}{1+\sigma} r^2 \Bigl(I + \Bigl( \frac{\sigma^2}{(1+\sigma)^2}
      - 2  \frac {\sigma}{1+\sigma}\Bigr)P \Bigr).
    \end{equation}

    \begin{remark}\rm
This is in agreement with formula (7.61) in
  \cite{s}.
\end{remark}

\begin{example}\label{ex:discrete}
In the case
$d\mu(v,r) = \frac 12\bigl(\delta_{v_0}(dv) + \delta_{-v_0}(dv)\bigr) \delta_{a}(dr)$
we have $\sigma = \rho a$ and $\pi = 0$.
Then,
\begin{gather}
  P\varphi (x) = \frac {\rho a}{2(1+\rho a)} \left(\varphi(x,v_0) +  \varphi(x,-v_0)\right)\\
  \label{46}
  C\varphi(x, \pm v_0) 
  = \frac{2 + \rho a}{2(1+\rho a)} \varphi(x,\pm v_0) - \frac {\rho a}{2(1+\rho a)}\varphi(x,\mp v_0).
\end{gather}
\end{example}

\subsection{Equilibrium fluctuations in the Euler scaling}
\label{ref}

Recall that the interparticle distance is of order $\ve$, i.e. the coordinates $(x,v,r)$
  are already rescaled in the Euler scale.
Let $X^\ve_t$ denote the free gas configuration at time $t$:
\begin{align}
  \label{xet1}
  X^\ve_t:=\{(x+vt,v,r):(x,v,t)\in X^\ve\}.
\end{align}
{For any $(x,v)\in \mathbb R^2$ and $t\ge 0$} define the flow
  \begin{equation}
  \label{jet1}
  \begin{split}
    j^\ve(x,v,t) &:= \ve\sum_{(\tx ,\tv ,\tr)\in X^\ve} \tr
    \bigl(1_{[\tv <v]} 1_{[x< \tx <x+(v-\tv )t]} - 1_{[\tv >v]}1_{[x+(v-\tv )t<\tx <x]}\bigr).
    \end{split}
  \end{equation}
  For any fixed $(x,v)$, the expectation of $ j^\ve(x,v,t)$ is given by
   \begin{align}
   j(x,v,t) &:= \EE j^\ve(x,v,t)   \\
    &= \iiint \rho\; dx\; d\mu(\tv ,\tr) \,\tr\,
      \bigl(1_{[\tv <v]} 1_{[x< \tx <x+(v-\tv )t]} - 1_{[\tv >v]}1_{[x+(v-\tv )t<\tx <x]}\bigr)\\
     &=  \rho \int  \tr \int_v^{+\infty} (v-\tv )\; t\; d\mu(\tv ,\tr) +
   \rho \int  \tr \int_{-\infty}^v (v-\tv )\; t\; d\mu(\tv ,\tr)\\
   &= t v\sigma - t \pi. \label{Ejet1}
 \end{align}

  We have the following limit as a consequence of the law of large numbers
  \begin{align}\label{eq:llnj}
   j^\ve(x,v,t)  \ \toas  j(x,v,t) & = t(v\sigma - \pi).
\end{align}

  \begin{figure}[htb]
\begin{center}
  \includegraphics[width=0.5\textwidth]{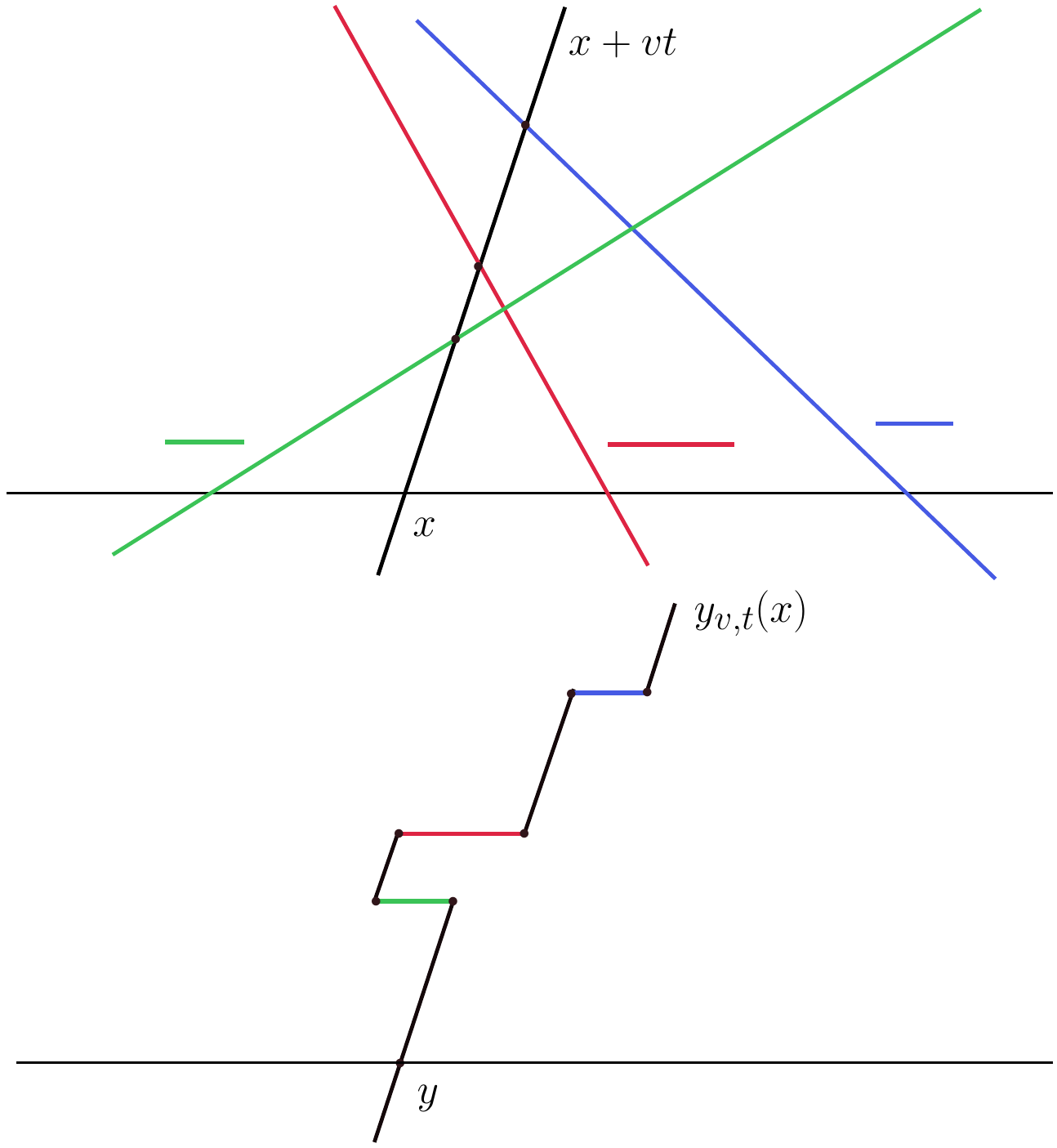}
  \caption{Upper figure: The ideal gas trajectory $(x+vt)_{t\in\RR}$ and the trajectories of the other particles, whose associated lengths are positive and represented by horizontal segments. Bottom: the associated quasi-particle trajectory $y_{v,t}(x)$.\label{fig:qs}}
\end{center}
\end{figure}

For any $(x,v)\in \mathbb R^2$ and $t\ge 0$, we define
  \begin{align}
    y^\ve_{v,t}(x) &:= x+ m_0^x(X^\ve) + vt+ j^\ve(x,v,t),   \label{yet1}
  \end{align}
see Fig. \ref{fig:qs}; its average is
  \begin{align}
    y_{v,t}(x) &:= \EE y^\ve_{v,t}(x) 
                 = x(1+\sigma) + vt +j(x,v,t)= x(1+\sigma) + t v(1+\sigma) - t\pi.
                 \label{y1}
  \end{align}

  {We say that  $y^\ve_{v,t}(x)$ is the position at time $t$ of the 
    \emph{quasi-particle} with velocity $v$ starting from $x$ (see Figure \ref{fig:qs}).
    Notice that this definition is valid for all $x,v$, it does not require that $(x,v)$
    are points in $X^\ve$.}

By \eqref{eq:llnj}
 \begin{equation}
   \label{22}
    y^\ve_{v,t} - y  \ \toas\  
    v^{\eff}(v) t,
 \end{equation}
 where the effective velocity is given by
 \begin{equation}
   \label{23}
 v^{\eff}(v) := v(1+\sigma) - \pi.
 \end{equation}

We have that, by \eqref{xlln}, the free gas empirical length measure  at time $t$ satisfies
  \begin{align}
    \label{net1}
    N^\ve_t\varphi := \ve \sum_{(x,v,r)\in X^\ve}    r\varphi( x+ vt, v, r)\toas
   \rho\iiint r\varphi(x+vt,v,r)\; dx \; d\mu(v,r) = \lang \varphi \rang.
  \end{align}

  The  $X$-fluctuation field at time $t$ satisfies
\begin{align}
\label{20}
  \xi^{X,\ve}_t (\varphi) &:=  \ve^{-1/2}\bigl(N^\ve_t\varphi  - \lang \varphi \rang\bigr)\tolaw \xi^X(\varphi_t),
\end{align}
where
 \begin{align}
  \varphi_t(x,v,r) &:= \varphi(x +tv,v,r). \label{ft56}
 \end{align}
The hard rod configuration and empirical measure at time $t$ are given by
\begin{align}
  Y^\ve_t &:= \bigl\{({y^\ve_{v,t}(x)},v,r):(x,v,r)\in X^\ve\bigr\},  \label{yy}\\
   K^\ve_t\varphi &:= \ve \sum_{(y,v,r)\in Y^\ve_t}    r\varphi( y, v, r)=
    \ve \sum_{(x,v,r)\in X^\ve} r\varphi({y^\ve_{v,t}(x)}, v, r).
   \label{Yem}
 \end{align}
 Using \eqref{115k} and similar argument used at the beginning of section \ref{sec:static-clt},
 we have the LLN for $K^\ve_t$:
\begin{align}
  \label{lnn0}
  K^\ve_t\varphi \mathop{\longrightarrow}_{\ve \to 0}&\; \rho \iiint r\varphi(y_{v,t}(x),v,r) dx\; d\mu(v,r)\\
  = &\;\rho \iint r \Bigl(\int \varphi\bigl(x(1+\sigma) + v^{\eff}(v) t,v,r\bigr) dx\Bigr) \; d\mu(v,r) \\
              = &\;\frac{\rho}{1 + \sigma}  \iiint r\varphi(y,v,r) dy\; d\mu(v,r)%
    =\frac{\rho}{1 + \sigma} \lang\varphi\rang.
\end{align}

We define the $Y$-fluctuation field at time $t$ by 
\begin{equation}
  \label{73t}
  \xi^{Y,\ve}_t(\varphi) := \ve^{-1/2}\bigl(K^\ve_t\varphi - \EE K^\ve_t\varphi\bigr).
\end{equation}

In order to shorten notations in the following we set $\pi = 0$,
the extension to the case $\pi\neq 0$ is straightforward.
We have
\begin{align}
  K^\ve_t\varphi - \EE K^\ve_t\varphi
  =  (K^\ve_t\varphi - A^\ve_t\varphi)
  + (A^\ve_t\varphi - \EE A^\ve_t\varphi)
  - (\EE K^\ve_t\varphi -\EE A^\ve_t\varphi). \label{117}
\end{align}
where 
\begin{align}
  A^\ve_t\varphi&:= \ve\sum_{(x,v,r)\in X^\ve} r\varphi\bigl((x+vt)(1+\sigma), v, r\bigr).%
\end{align}
The last term in \eqref{117} gives
\begin{align}
  &\hspace{-4mm} \ve^{-1/2} (\EE K^\ve_t\ \varphi -\EE A^\ve_t\ \varphi)\\
  &=  \mathbb E\Bigl( \ve^{1/2} \sum_{(x,v,r)\in X^\ve}
    r\bigl[\varphi\bigl( x +  m_0^x(X^\ve)+ vt + j^\ve(x,v,t), v, r\bigr)
   \label{120}\\
  &\qquad\qquad\qquad\qquad - \varphi\bigl( (x+vt)(1+\sigma), v, r\bigr)\bigr]\Bigr)\\
  &=  \mathbb E\Bigl( \ve^{1/2} \sum_{(x,v,r)\in X^\ve}
    r(\partial_y \varphi)\bigl( (x+vt)(1+\sigma), v, r\bigr)\label{890}\\
  &\qquad\qquad\qquad\qquad \times \bigl(m_0^x(X^\ve_t) +  j^\ve(x,v,t) - (x+ vt) \sigma\bigr)
     \Bigr) + R^\ve_t \\
&=  \iiint  r(\partial_y \varphi)\bigl( (x+vt)(1+\sigma), v, r\bigr) \label{123}
 \ve^{-1/2} \\
  &\qquad\qquad\qquad\qquad \times \mathbb E\bigl[m_0^x(X^\ve)  +  j^\ve(x,v,t) - (x+ vt) \sigma\bigr] dx\; d\mu(v,r) + R^\ve_t\\
  &= R^\ve_t,\label{124}
\end{align}
where $R^\ve_t$ is of smaller order {in $\ve$.
 The identity \eqref{123}
  follows from the Slyvniak-Mecke formula and from \eqref{mab1} and \eqref{Ejet1}.}
By \eqref{xclts} the second term in \eqref{117} gives
\begin{align}
  \notag
  &\vspace {-5mm}\ve^{1/2}(A^\ve_t-\EE A^\ve_t)
     \tolaw   \xi^X_t(\ckphi_t), \label{aea1}
\end{align}
where
\begin{equation}
  \label{9}
 \ckphi_t(x,v,r) := \varphi\bigl((x+vt) (1+ \sigma) ,v,r\bigr).
\end{equation}
Finally, the first term in \eqref{117} gives 
\begin{align}
  &\hspace{-4mm}\ve^{1/2}(K^\ve_t\varphi - A^\ve_t\varphi)\\
  & = \ve^{1/2}\sum_{(x,v,r)\in X^\ve}  r \bigl[\varphi\bigl(  x+ m_0^x(X^\ve) + vt+ j^\ve(x,v,t), v, r\bigr)
   - \varphi\bigl( (x+vt) (1+\sigma), v, r\bigr) \bigr]\\
  & = \ve^{1/2}\sum_{(x,v,r)\in X^\ve}
    r( \partial_x\varphi)\bigl((x+vt) (1+\sigma), v, r\bigr)
    (m_0^x(X^\ve)+  j^\ve(x,v,t) - \sigma (x+vt)) + R^\ve_t
  \\
  & = \ve\sum_{(x,v,r)\in X^\ve}  r (\partial_x\varphi)\bigl(  (x+vt) (1+\sigma), v, r\bigr)
    \ve^{-1/2} (m_0^x(X^\ve) + j^\ve(x,v,t) - \sigma (x+vt))  + R^\ve_t,
    \label{11}
\end{align}
where $R^\ve_t$ is smaller order.
It can be easily shown that
\begin{equation}
  \label{jet1}
  \begin{split}
  \ve^{-1/2}  j^\ve(0,0,t) &:= \ve^{1/2} \sum_{(\tx ,\tv ,\tr)\in X^\ve} \tr
    \bigl(1_{[\tv <0]} 1_{[0< \tx < -\tv t]} - 1_{[\tv >0]}1_{[-\tv t<\tx <0]}\bigr)
    \end{split}
  \end{equation}
  converges in law to a Wiener process $\tilde B(t,0)$ with a finite variance.
  By applying Lemma \ref{LemmaPP} we have
  \begin{equation}
    \label{eq:22}
    \begin{split}
    &\ve\sum_{(x,v,r)\in X^\ve}  r (\partial_x\varphi)\bigl(  (x+vt) (1+\sigma), v, r\bigr)
    \ve^{-1/2}  j^\ve(0,0,t)\\
    &\tolaw  \tilde B(t,0)
    \iint d\mu(v,r)  r \int (\partial_x\varphi)\bigl(  (x+vt) (1+\sigma), v, r\bigr) dx = 0.
\end{split}
\end{equation}
Consequently we can write \eqref{11} as
\begin{align}
  \label{eq:23}
  \begin{split}
    &\ve\sum_{(x,v,r)\in X^\ve}  r (\partial_x\varphi)\bigl(  (x+vt) (1+\sigma), v, r\bigr) \\
  &\qquad \qquad \times\,  \ve^{-1/2} (m_0^x(X^\ve) + j^\ve(x,v,t) -  j^\ve(0,0,t) - \sigma (x+vt))
    + R^\ve_t.
  \end{split}
\end{align}
Remark that
\begin{equation}
  \label{eq:24}
  m_0^x(X^\ve) + j^\ve(x,v,t) =  j^\ve(0,0,t) + m_0^{x+vt}(X_t^\ve) ,
\end{equation}
and the first term in \eqref{eq:23} can be written as
\begin{equation}
  \label{eq:25}
  \ve\sum_{(x,v,r)\in X^\ve}  r (\partial_x\varphi)\bigl(  (x+vt) (1+\sigma), v, r\bigr)
    \ve^{-1/2} (m_0^{x+vt}(X_t^\ve) - \sigma (x+vt)).
  \end{equation}
  By \eqref{20} we have
  \begin{equation}
    \ve^{-1/2} (m_0^{x}(X_t^\ve) - \sigma (x)) \tolaw
   \xi_t^X\bigl(1_{[0,x]} 1_{[x>0]} - 1_{[x,0]} 1_{[x<0]}\bigr) =:  \bar B(t,x), 
   \label{eq:26}
 \end{equation}
 where, for any $t\ge 0$, $\bar B(t,x)$ is a double sided Wiener process in $x\in \mathbb R$. 
Then, by Lemma \ref{LemmaPP}, \eqref{eq:25} converges in law to
\begin{equation}
  \label{eq:27}
  \begin{split}
  & {\rho}\iiint r (\partial_x\varphi) \bigl((x+vt) (1+\sigma), v, r\bigr) \, \bar B(t, x+vt)\,
   dx \; d\mu(v,r)\\
  & = {\rho}\iiint r (\partial_x\varphi) \bigl(x (1+\sigma), v, r\bigr) \, \bar B(t, x)\,
   dx \; d\mu(v,r)\\
   &=\frac \rho{1+\sigma}  \iiint\xi_t^X\bigl(1_{[0,x]} 1_{[x>0]} - 1_{[x,0]} 1_{[x<0]}\bigr)\,
       r\,\partial_x   \ckphi(x, v, r) \;dx\, d\mu(v,r)\\
   &= - \frac {1}{1+\sigma}\,
   \xi^X_t\Bigl(\rho\, \iint r \ckphi(\cdot , \tv , \tr) \; d\mu(\tv ,\tr)\Bigr)
   = -  \xi^X_t\Bigl( P \ckphi \Bigr). 
 \end{split}
\end{equation}

{

We first make explicit the change of variables between the fluctuation fields
$\xi^X$ and $\xi^Y$. Define
\begin{equation}
  \label{eq:P}
  (\tilde P\psi)(x,v,r)
  = \frac{\rho}{1+\sigma}\int r'\psi(x,v',r')\,d\mu(v',r'),
\end{equation}
and let
\begin{equation}
  (T\varphi)(x,v,r)=\varphi((1+\sigma)x,v,r).
\end{equation}
The operator which maps test functions for the
hard-rod field to test functions for the free field is therefore
\begin{equation}
  \label{eq:Cdef}
  C=(I-\tilde P)T.
\end{equation}
Notice that, since $ \sigma=\rho\int r\,d\mu(v,r)$,
we have
\begin{equation}
  \label{eq:P2}
  \tilde P^2=\frac{\sigma}{1+\sigma}\tilde P.
\end{equation}
Consequently,
\begin{equation}
  \label{eq:IminusPinverse}
  (I-\tilde P)^{-1}=I+(1+\sigma) \tilde P,
\end{equation}
as is immediately checked using \eqref{eq:P2}. Since
$T^{-1}\psi(q,v,r)=\psi(q/(1+\sigma),v,r)$, it follows from
\eqref{eq:Cdef} that
\begin{equation}
  \label{eq:Cinv}
  \begin{split}
  (C^{-1}\psi)(q,v,r)
  &=\psi\left(\frac{q}{1+\sigma},v,r\right)
  +\rho\int r'\psi\left(\frac{q}{1+\sigma},v',r'\right)
  \,d\mu(v',r').
  \end{split}
\end{equation}

Let
\begin{equation}
  \label{eq:D}
  D\psi(x,v,r)=v\partial_x\psi(x,v,r).
\end{equation}
Then the argument leading to (2.30) gives
\begin{equation}
  \label{eq:evolution}
  \begin{split}
  \xi_t^{Y,\ve}(\varphi)\tolaw \xi_t^Y(\varphi)
  &:=\xi_t^X(C\varphi)
   =\xi_0^X(e^{tD}C\varphi)\\
  &=\xi_0^Y(C^{-1}e^{tD}C\varphi)
   =\xi_0^Y(e^{tA^*}\varphi),
  \end{split}
\end{equation}
where
\begin{equation}
  \label{eq:semigroup}
  e^{tA^*}=C^{-1}e^{tD}C.
\end{equation}
Thus $e^{tA^*}$ is a semigroup with generator
\begin{equation}
  \label{eq:Astar}
  A^*=C^{-1}DC.
\end{equation}

We now compute this generator explicitly. From \eqref{eq:Cdef},
\begin{equation*}
  (C\varphi)(x,v,r)
  =\varphi((1+\sigma)x,v,r)
  -\frac{\rho}{1+\sigma}\int
  r'\varphi((1+\sigma)x,v',r')\,d\mu(v',r'),
\end{equation*}
and hence
\begin{equation*}
  \begin{split}
  (DC\varphi)(x,v,r)
  ={}&(1+\sigma)v\,\partial_q\varphi((1+\sigma)x,v,r)\\
  &-\rho v\int r'\partial_q\varphi((1+\sigma)x,v',r')
  \,d\mu(v',r').
  \end{split}
\end{equation*}

Recall that we assume here $\pi =0$ for simplicity, then using \eqref{eq:Cinv} we obtain
\begin{equation}
  \label{eq:Astar-explicit}
  \begin{split}
  (A^*\varphi)(q,v,r)
  ={}&(1+\sigma)v\,\partial_q\varphi(q,v,r)\\
  &+\rho\int r'\big((1+\sigma)v'-v\big)
  \partial_q\varphi(q,v',r')\,d\mu(v',r').
  \end{split}
\end{equation}

Taking the transpose of \eqref{eq:Astar-explicit} with respect to
$dq\,d\mu(v,r)$ gives
\begin{equation}
  \label{eq:A-explicit}
  \begin{split}
  (A\varphi)(q,v,r)
  ={}&-(1+\sigma)v\,\partial_q\varphi(q,v,r)\\
  &-\rho r\int
  \big((1+\sigma)v-v'\big)
  \partial_q\varphi(q,v',r')\,d\mu(v',r').
  \end{split}
\end{equation}
Indeed, this follows directly by interchanging $(v,r)$ and $(v',r')$
in the integral term and integrating by parts in $q$.
In particular,
\begin{equation}
  \label{eq:weak}
  \partial_t\xi_t^Y(\varphi)=\xi_t^Y(A^*\varphi).
\end{equation}
Equivalently, in the weak sense,
\begin{equation}
  \label{eq:linearized}
  \partial_t\xi_t^Y-A\xi_t^Y=0.
\end{equation}
The operator \eqref{eq:A-explicit} is precisely the linearization, around the
homogeneous equilibrium state and under the assumption $\pi=0$, of the Euler
equation given in Section 4. Convergence to this linearized Euler equation was
already proved in \cite{boldrighini-wick}.
}

\section{Equilibrium Fluctuations in the diffusive scaling}
\label{sec:equil-fluct-diff}
\subsection{ Quasi-particles in the diffusing scaling}
\label{sec:tagg-rods-diff}

Recall the definition of $j^\ve(x,v, t)$ given by \eqref{jet1} for any $(x,v)\in \mathbb R^2$.
and the corresponding convergence \eqref{eq:llnj}. We look now at this flow in a larger,
diffusive, time scale, and we have that
\begin{equation}
  \label{80j}
j^\ve(x,v,\ve^{-1}t) - \bigl(v\sigma-\pi\bigr) \ve^{-1} t
  \ \tolaw
   \ B_t(v).
 \end{equation}
 where $\bigl(B_t(v)\bigr)_{v\in\R} =: \mathbb B_t$
 is a family of Brownian motions in $t$ with
 covariance given by \eqref{eq:13} and initial value $B_0(v)=0$ for all $v$.
 Notice in \eqref{80j} that $B_t(v)$ does not depend on $x$.
 In fact we will see below that, given two points $x$ and $\tx$,
 the variance of the difference $j^\ve_{v,t}(x,v,\ve^{-1}t)-j^\ve_{v,t}(\tx,v,\ve^{-1}t)$
 is of order $2\ve |x-\tx|$, as there are around $\ve^{-1} |x-\tx|$
 particles between $x$ and $\tx$, and between $x+vt$ and $\tx+vt$,
 and each of those particles contributes $\ve^2$ to the variance, see \eqref{eq:84}.
 This explains why the limit in \eqref{80j} does not depend on $x$. 

 Given a point $(x,v) \in \mathbb R^2$ recall that $y^\ve_{v,t}(x)$ is defined by \eqref{yet1}.
As a consequence of \eqref{80j} we have
\begin{equation}
  \label{eq:79}
  \begin{split}
y^\ve_{v,\ve^{-1}t}(x) - v^{\eff}(v) \ve^{-1} t \ \tolaw
   \ y + B_t(v),    
  \end{split}
\end{equation}
where $y=x(1+\sigma)$.
Observe that
\begin{align}
  \label{800}
  y^\ve_{v,\ve^{-1}t}(x)  - v^{\eff}(v) \ve^{-1} t  =
  x + m^x_0(X^\ve) - (1+\sigma)x +  j^\ve(x,v,\ve^{-1}t) - \bigl(v\sigma-\pi\bigr) \ve^{-1} t.
\end{align}
Since
$m^x_0\displaystyle{\toas} \sigma x$, the limit in \eqref{eq:79} is equivalent to  \eqref{80j}

We compute explicitly the covariances.
By \eqref{yet1} 
 Observe that
 \begin{equation}
   \label{82}
   \ve \sum_{(\tx,\tv,\tr )\in X^{\ve}} \tr  1_{[\tv<v]} 1_{[x< \tx<x+(v-\tv)\ve^{-1}t]}
   = \ve \sum_{(\tx,\tv,\tr )\in X^{\ve^{2}}} \tr  1_{[\tv<v]} 1_{[\ve x< \tx<\ve x+(v-\tv)t]},
 \end{equation}
 where $X^{\ve^{2}}:=\{(\ve \tx,\tv,\tr): (\tx,\tv,\tr)\in X^\ve\}$,
 is obtained from $X^{\ve}$ by rescaling all positions by a factor
 $\ve$, so that $X^{\ve^{2}}$ is a Poisson process of intensity measure
 $\ve^{-2}\rho\, d\mu(v,r)$.
We have that
 \begin{equation}
   \label{eq:87}
   \mathbb E\Bigl(\ve^2 \sum_{(x',v',r')\in X^{\ve^{2}}}
     r' 1_{[v'<v]} 1_{[\ve x< x'<\ve x+(v-v')t]}\Bigr) =
  t \rho\int r \int_{-\infty}^v (v-v') d\mu(v',r).
\end{equation}
{
  Define the fluctuation field
  \begin{equation}
    \label{eq:29}
    \tilde\xi^\ve(\varphi) = \ve \sum_{(\tx,\tv,\tr )\in X^{\ve^{2}}} \varphi(x,v,r)
    - \ve^{-1} \lang \varphi \rang \\
    = \ve \sum_{(\tx,\tv,\tr )\in X^{\ve}} \varphi(\ve x,v,r) - \ve^{-1} \lang \varphi \rang.
  \end{equation}
  As for $\xi^\ve$, cf.  \eqref{xclts}, $\tilde\xi^\ve$ converge to a centered Gaussian field
   $\tilde\xi$ with
   the same covariance as $\xi$, cf. \eqref{4}. But  $\tilde\xi$ and  $\xi$ are independent.
   In fact consider two test function $\varphi, \psi$ such that  $\lang \varphi \rang=0$ and
   $\lang \psi \rang=0$:
   \begin{equation}
     \label{eq:30}
     \mathbb E\Bigl( \tilde\xi^\ve(\varphi) \xi^\ve(\psi)   \Bigr) =
     \ve^{1/2} \iint d\mu(v,r) r^2 \int \varphi(\ve x,v,r) \psi(x,v,r) \; dx
     \mathop{\longrightarrow}_{\ve \to 0} 0.
   \end{equation}

 Consider the function
\begin{equation}
  \label{eq:81o}
   \varphi_{\ve x,v,t}(x',v') = 1_{[v'<v]} 1_{[\ve x< x'<\ve x+(v-v')t]}
   -1_{[v'>v]}1_{[\ve x+(v-v')t<x'<\ve x]},
 \end{equation}
 where $(x',v',r') \in X^{\ve^{2}}$.
 Then
 \begin{equation}
   \label{eq:31}
     j_{X^\ve}(x,v,\ve^{-1}t) - (v\sigma-\pi) \ve^{-1} t = \tilde\xi^{X,\ve}(\varphi_{\ve x,v,t})
 \end{equation}
For $x=0$ we have
 \begin{equation}
   \label{eq:81}
   \begin{split}
     j_{X^\ve}(0,v,\ve^{-1}t) - (v\sigma-\pi) \ve^{-1} t
     \ \tolaw\ \tilde\xi^X(\varphi_{0,v,t}),
   \end{split}
 \end{equation}
 which has variance
 \begin{align}
   \label{eq:83}
     &  \rho \iiint   \tilde r^2
       \bigl(1_{[\tilde v<v]} 1_{[0<\tilde  x<(v-\tilde v)t]} +
         1_{[\tilde v> v]}1_{[(v-\tilde v)t<\tilde x<0]}\bigr)
   d\tilde x \; d\mu(\tilde v,\tilde r)\\
   &=  \rho \int  \tilde r^2\int^v_{-\infty} (v- \tilde v)t d\mu(\tilde v, \tilde r) +
   \rho \int  \tilde r^2 \int^{\infty}_v (\tilde v-v)t d\mu(\tilde v,\tilde r)\\
   &= t \rho\iint  \tilde r^2 |v-\tilde v| d\mu(\tilde v,\tilde r) =: t \mathcal D(v).
 \end{align}
 Notice that $\mathcal D(v)>0$ except for the trivial
   $\mu$ concentrated on a single velocity.
}

 Computing the correlation for different initial position, assuming $x< \bar x$:
 \begin{align}
   \label{eq:84}
   & \mathbb E\bigl( \xi^{X,\ve}(\varphi_{\ve x,v,t}) \xi^{X,\ve}(\varphi_{\ve\bar x,v,t})\bigr)\\
   &=
     \rho \iiint   r^2  \bigl(1_{[v'<v]} 1_{[\ve x< x'<\ve x+(v-v')t]}
       - 1_{[v'>v]}1_{[\ve x+(v-v')t<x'<\ve x]}\bigr)\\
    &\qquad\times \bigl(1_{[v'<v]} 1_{[\ve \bar x< x'<\ve \bar x+(v-v')t]} -
       1_{[v'>v]}1_{[\ve\bar x+(v-v')t<x'<\ve \bar x]}\bigr)
     dx' \; d\mu(v',r)\\
     &=\rho \iiint   r^2  \Big(1_{[v'<v]} 1_{[\ve x< x'<\ve x+(v-v')t]}
     1_{[\ve \bar x< x'<\ve \bar x+(v-v')t]}\\
     & \qquad\qquad + 1_{[v'>v]}1_{[\ve x+(v-v')t<x'<\ve x]}1_{[\ve\bar x+(v-v')t<x'<\ve\bar x]}\Big)
       dx' \; d\mu(v',r)\\
      & = \rho \int r^2 \int_{-\infty}^{v-\ve (\bar x - x)/t}\bigl[ (v-v')t -\ve (\bar x - x)\bigr]
       d\mu(v',r)\\
      & \qquad\qquad+ \rho \int r^2 \int_{v+\ve (\bar x - x)/t}^{+\infty} (\ve(x-\bar x) -(v-v')t) d\mu(v',r)\\
    = & \rho \bigl[ vt -\ve (\bar x - x)\bigr] \int r^2 \int_{-\infty}^{v- \ve (\bar x - x)/t} d\mu(v',r)
     - \rho t \int r^2 \int_{-\infty}^{v-\ve (\bar x - x)/t} v' d\mu(v',r)\\
   & + \rho \bigl[ \ve(x -\bar x) - vt \bigr] \int r^2 \int_{v+\ve (\bar x - x)/t}^{+\infty} d\mu(v',r)
     + \rho t \int r^2 \int_{v+\ve (\bar x - x)/t}^{+\infty} v' d\mu(v',r)\\
       = &\rho \bigl[ vt -\ve (\bar x - x)\bigr] \int r^2 \int_{-\infty}^{v- \ve (\bar x - x)/t} d\mu(v',r)
       - \rho t \int r^2 \int_{-\infty}^{v- \ve(\bar x - x)/t} v' d\mu(v',r)\\
   &- \rho \bigl[ vt + \ve(\bar x - x) \bigr] \int r^2 \int_{v+\ve(\bar x - x)/t}^{+\infty} d\mu(v',r)
     + \rho t \int r^2 \int_{v+\ve(\bar x - x)/t}^{+\infty} v' d\mu(v',r) \\
       &\qquad
       \mathop{\longrightarrow}_{\ve \to 0} t \rho\iint  (r')^2 |v-v'| d\mu(v',r') = t \mathcal D(v).
 \end{align}
 It follows from \eqref{eq:84} that
 two tagged quasi-particles
 with the same velocity are asymptotically completely correlated
 {and
 \begin{equation}
   \label{eq:33}
   \lim_{\ve\to 0}
   \mathbb E\left( \left|\tilde\xi^{X,\ve}(\varphi_{\ve x,v,t}) - \tilde\xi^{X,\ve}(\varphi_{0,v,t})\right|^2\right)
   =0.
 \end{equation}
}
 Considering the correlation at different velocities $v< \bar v$:
 \begin{align}
   \label{eq:88}
   &\mathbb E\bigl( \xi^{X,\ve}(\varphi_{\ve x,v,t}) \xi^{X,\ve}(\varphi_{\ve x,\bar v,t})\bigr)\\
   &=
     \rho \iiint   r^2  \bigl(1_{[v'<v]} 1_{[\ve x< x'<\ve x+(v-v')t]}
       - 1_{[v'>v]}1_{[\ve x+(v-v')t<x'<\ve x]}\bigr)\\
     &\qquad\times\bigl(1_{[v'<\bar v]} 1_{[\ve x< x'< \ve x+(\bar v-v')t]} -
       1_{[v'>\bar v]}1_{[\ve x+(\bar v-v')t<x'<\ve x]}\bigr)
     \; dx' \; d\mu(v',r)\\
   &= \rho \iiint   r^2
     \Big(1_{[v'<v]} 1_{[\ve x< x'<\ve x+(v-v')t]}1_{[\ve x< x'< \ve x+(\bar v-v')t]}\\
   & \qquad\qquad\qquad+
     1_{[v'>\bar v]}1_{[\ve x+(\bar v-v')t<x'<\ve x]}1_{[\ve x+(v-v')t<x'<\ve x]}\\
   & \qquad\qquad\qquad
     - 1_{[v<v'<\bar v]}1_{[\ve x+(v-v')t<x'<\ve x]}1_{[\ve x< x'< \ve x+(\bar v-v')t]} \Big)
       \; dx' \; d\mu(v',r)\\
      &   = \rho \iiint   r^2  \Big(1_{[v'<v]} 1_{[\ve x< x'<\ve x+(v-v')t]}
       + 1_{[v'>\bar v]}1_{[\ve x+(\bar v-v')t<x'<\ve x]}\Big) \; dx' \; d\mu(v',r)\\
   & = t \rho \iint r^2 \bigl(1_{[v'<v]} (v-v') + 1_{[v'>\bar v]} (v' - \bar v)
     \bigr) d\mu(v',r)\\
       &=: t \Gamma(v,\bar v).
 \end{align}

 Noting that
 \begin{equation*}
   1_{[v'<v]} (v-v') + 1_{[v'>\bar v]} (v' - \bar v) =
   \frac 12 \bigl( |v-v'| + |v' - \bar v| - (\bar v - v)\bigr), 
 \end{equation*}
 we have that 
 \begin{equation}
   \label{eq:13}
   \Gamma(v,\bar v) =  \frac{1}{2}\Bigl( \mathcal D(v) + \mathcal D(\bar v)
   - (\bar v-v) \rho \iint r^2 d\mu(v',r)\Bigr).
 \end{equation}

 \begin{remark}\rm
The expression \eqref{eq:13} for the covariance corresponds to the L{\'e}vy
Chentsov field, see details in \cite{ffgs22}.
   
  From \eqref{eq:88} we can see immediately that $\Gamma(v,\bar v) \ge 0$.
  In the particular case where there are only two velocities admitted,
  for example
  $d\mu(v,r) = \delta_a(dr) \frac 12 \left(\delta_{-1}(dv)+ \delta_1(dv)\right) $,
  we have $\mathcal D(\pm 1) = \rho a^2$ and $\Gamma(1, -1) = 0$.
  Typically this decorrelation
  happens only when two velocities at most are present.
\end{remark}

\subsection{Diffusive evolution of density fluctuations}
\label{sec:diff-evol-dens}

Define the fluctuation field at diffusive scaling as
\begin{equation}
   \label{21}
   \begin{split}
     \Xi^{Y,\ve}_t (\varphi) :=  \ve^{-1/2}\Bigl[\ve \sum_{(x,v,r)\in X^\ve}
       r\varphi\bigl[ y^\ve_{v,\ve^{-1}t}(x) - v^{\eff}(v) \ve^{-1} t , v, r\bigr]-
       \frac 1{1+\sigma} \lang \varphi \rang\Bigr].
   \end{split}
 \end{equation}
 Notice that this is recentered on the Euler characteristics. 

 Define \[\varphi_{\mathbb B_t}(y,v,r) := \varphi\bigl(y +B_t(v),v,r\bigr).\]
   The next theorem contains our main result.
\begin{theorem}
  \label{thm1}
 \begin{equation}
   \label{32}
   \Xi^{Y,\ve}_t (\varphi) \tolaw \  \Xi^{Y}_t (\varphi)
   = \Xi_0^{Y} \bigl(\varphi_{\mathbb B_t})\bigr),
 \end{equation}
 where $\mathbb B_t = (B_t(v))_{v\in \R}$ is a family of Wiener processes with covariance
 \begin{equation}
   \label{eq:14}
   \bbE \bigl(B_t(v), B_t(w)\bigr) = {t \Gamma(v,w)}
 \end{equation}
 with $\Gamma$ defined in \eqref{eq:88}.
 Furthermore $\mathbb B_t$ is independent from $\Xi_0^{Y}$.
  More formally
  \begin{equation}
    \label{27}
    \begin{split}
    \Xi^{Y}_t (\varphi) &= \iiint r\varphi\bigl(y+ B_t(v),v,r\bigr)
    d\xi_0^Y(y,v,r)\\
    &=  \iiint r C\varphi\bigl((1+\sigma) x + B_t(v) ,v,r\bigr) d\xi_0^X(x,v,r).
  \end{split}
  \end{equation}
\end{theorem}

  We prove this theorem after some comments. By \eqref{27}, $\Xi^{Y}_t$ solves the stochastic differential equation
  \begin{equation}
    \label{33}
    d \Xi^{Y}_t(\varphi) = \frac 1{2}  \Xi^{Y}_t(\mathcal D \partial_y^2 \varphi) dt
    - \iiint  (\partial_y \varphi)(y,v,r) dB_t(v) d\Xi^Y_t(y,v,r),
  \end{equation}
  or in the time integrated form:
  \begin{align}
    \label{33i}
    \Xi^{Y}_t(\varphi) &= \Xi^{Y}_0(\varphi)+  \int_0^t
    \frac 1{2}  \Xi^{Y}_s(\mathcal D \partial_y^2 \varphi) ds\\
    &\qquad\qquad- \int_0^t \iiint  (\partial_y \varphi)(y,v,r)
    dB_s(v) d\Xi^Y_s(y,v,r)\\
    &= \Xi^{Y}_0(\varphi)+  \int_0^t
      \frac 1{2}  \Xi^{Y}_s(\mathcal D \partial_y^2 \varphi) ds-
      \int_0^t \Xi^Y_s\bigl(\;  \partial_y \varphi\;  dB_s\bigr),
  \end{align}
  where the last term is a martingale with quadratic variation
  \begin{align}
    \label{38}
    \int_0^t
    \Bigl(\iiint \sqrt{\mathcal D(v)} (\partial_y \varphi)(y,v,r)
      d\Xi^Y_s(y,v,r)\Bigr)^2 ds
   = \int_0^t \Xi^Y_s\bigl(\sqrt{\mathcal D} (\partial_y \varphi)\bigr)^2 ds .
  \end{align}
  The expectation of the quadratic variation \eqref{38} is given by
  \begin{equation}
    \label{39}
    t  \frac{\rho}{1+\sigma}\iiint \mathcal D(v) r^2 (\partial_y \varphi)^2(y,v,r)
    dy \; d\mu(v,r).
  \end{equation}

  \begin{proof}[Proof of Theorem \ref{thm1}]
    In order to simplify notation we set $\pi = 0$, the extension to general $\pi\neq 0$ is straightforward.
    Denote
    \begin{equation}
      \label{eq:54}
      B_t^\ve(x,v) =  j^\ve(x,v,\ve^{-1}t)  - \sigma v \ve^{-1} t, \qquad
      {\mathbb B}^\ve_t = {\mathbb B}^\ve_t(X^\ve)= \{B_t^\ve(x,v), (x,v)\in \mathbb R^2\}.
    \end{equation}
  
    By \eqref{80j} ${\mathbb B}^\ve_t$ converges in law to the field of Brownian motions
    \begin{equation}
      \label{eq:51}
      {\mathbb B}^\ve_t \tolaw  {\mathbb B}_t = (B_t(v))_{v\in\R}.
\end{equation}    
   Then we can rewrite 
\begin{equation}
  \label{eq:49}
  \begin{split}
    & \Xi^{Y,\ve}_t (\varphi) = 
   \ve^{-1/2}\Bigl[\ve \sum_{(x,v,r)\in X^\ve} r
  \varphi\bigl[ x(1+\sigma) + B_t^\ve( x,v) , v, r\bigr]-
  \frac 1{1+\sigma} \lang \varphi \rang\Bigr]\\
  &+ \ve \sum_{(x,v,r)\in X^\ve} r
 (\partial_y \varphi)\bigl[ x(1+\sigma) + B_t^\ve( x,v) , v, r\bigr]
 \ve^{-1/2} \left(m_{0}^{x}(X_{0}^\ve) - \sigma x\right) + R^\ve_t,
     \end{split}
   \end{equation}
   where $R^\ve_t\to 0$ in law.

   We deal with the first term on the RHS of \eqref{eq:49}.
   
   By \eqref{xclts} and \eqref{eq:51},
   the couple $(\xi^{X,\ve},  {\mathbb B}^\ve_t)$ converges jointly in law
   to $(\xi^{X},  {\mathbb B}_t) $. By \eqref{eq:30} they are independent.
   By the mapping theorem (\cite{zbMATH04060392}, chapter 1, section 2)
   we have that the first term on the RHS of \eqref{eq:49} converges in law to
  { $ \xi^X(\widetilde\varphi_{t})$} where
   \begin{equation}
     \label{eq:50}
    { \widetilde\varphi_{t}} (x,v,r) = \varphi(x(1+\sigma)+B_t(v),v,r).
     \end{equation}
   About the second term of \eqref{eq:49}, 
   combining the argument above and the one used in \eqref{131s} to \eqref{135s}
   this converges in law to $- \xi^X(P \widetilde\varphi_{t})$,
   where $P$ is
   the operator defined in \eqref{15s}, i.e.
   \begin{equation}
     \label{eq:35}
     P \widetilde\varphi_{t}(x) = \frac{\rho}{1+\sigma}
     \iint r  \widetilde\varphi_{t} (x,v,r) d\mu(v,r) =
      \frac{\rho}{1+\sigma}
     \iint r \varphi(x(1+\sigma)+B_t(v),v,r) d\mu(v,r). 
   \end{equation}
   Then by \eqref{17}, recalling that $C= I - P$, and the above we have that
   \begin{equation}
     \label{eq:36}
     \Xi^{Y}_t (\varphi) := \lim_{\ve\to 0} \Xi^{Y,\ve}_t (\varphi) =
     \xi^X(\widetilde\varphi_{t}) -  \xi^X(P \widetilde\varphi_{t})
     = \xi^X(C\widetilde\varphi_{t}) = \xi^Y(\varphi_{\mathbb B_t}).
     \qedhere
   \end{equation}
  \end{proof}

     Formally, choosing $\varphi_{k,w}(x,v,r) = e^{i2\pi xk} \varphi(r) \delta(v-w)$,
     we have that
     \begin{align}
       \label{41}
       \hat\varphi(k,w,t) :=\Xi^{Y}_t (\varphi_{k,w}) &= \int d\xi^Y(y,r,v)
       e^{i2\pi k (y + B_t(w))} \varphi(r) \delta(v-w)\\
       &= \xi^Y\Bigl(e^{i2\pi k (\cdot + B_t(w))} \varphi(\cdot) \delta(\cdot-w)\Bigr)
     \end{align}
      satisfies the SDE
      \begin{equation}
        \label{42}
        \begin{split}
          d  \hat\varphi(k,w,t) = - \frac{(2\pi k)^2}{2}
          \mathcal D(w) \hat\varphi(k,w,t)
          + {i2\pi k} \hat\varphi(k,w,t) dB_t(w).
        \end{split}
      \end{equation}
      Notice that $|\hat\varphi(k,w,t)|^2 = |\hat\varphi(k,w,0)|^2$
      for any $k$, a persistence on the macroscopic scale of the complete integrability
      of the dynamics also at the level of these fluctuations.

      \begin{remark}
\rm        Since we consider also systems where lengths $r$ can be negative,
in the case that $\sigma = 0$ and $\pi = 0$ the macroscopic evolution
of fluctuations in the Euler scaling is the same as the independent point particles.
But in the diffusive scaling the fluctuations have non trivial behaviour.
      \end{remark}

\section{A remark about inhomogeneous
initial distribution}
\label{sec:inhom-non-stat}

Let $f_0(x,v,r)$ be a nice non-negative bounded function on $\mathbb R^3$  and
$X^\ve$ the Poisson process on $\mathbb R^3$ with intensity
$\ve^{-1} f_0(x,v,r) dx \;dv\; dr$.

In the Euler scaling, the empirical distribution of the free gas converges to the solution of
\begin{equation}
  \label{eq:15}
  \partial_t f_t(x,v,r) + v\partial_x f_t(x,v,r) = 0,
\end{equation}
with initial condition given by $f_0$.
For the rods density this corresponds to the equation
\begin{equation}
  \label{eq:16}
  \begin{split}
    \partial_t g_t(y,v,r) + \partial_y \bigl(v^{\eff}(y,v,t) g_t(y,v,r) \bigr) = 0,\\
    v^{\eff}(y,v,t) = v + \frac{\iint r (v-w) g_t(y,w,r) dw dr}{1 - \iint r g_t(y,w,r) dw dr},
\end{split}
\end{equation}
as proven in \cite{ffgs22}.

For \emph{generic} initial conditions, we can guess that, if
the initial density $f$ is absolutely continuous in the $x$ and $v$ coordinates, then it satisfies that the limit as $t\to\infty$ of
$f_t(x,v,r)$ is constant in $x$, that is,  $f_t(x,v,r) \longrightarrow \rho \bar f(v,r)$,
for some $\rho\in \R_+$ and $\bar f(v,r)$.
This suggests that in a diffusive time scale the system essentially behaves like if it is
a stationary state determined by a Poisson point field $\rho \bar f(v,r) dx\; dv\; dr$,
and the analysis for the macroscopic fluctuations of \autoref{sec:equil-fluct-diff}
applies.

\section{A limit theorem for Poisson process}
\label{sec:lemmaspp}

\begin{lemma}\label{LemmaPP}
Let $\varphi(x,v,r)$ a smooth function on $\mathbb R^3$
with compact support in $x$, {and $\{B^\ve(x), x\in \mathbb R\}$ a process
converging in law to $\{B(x), x\in \mathbb R\}$.} 
Then
\begin{equation}
  \label{eq:1}
  \lim_{\ve\to 0} \ve \sum_{(x,v,r)\in X^\ve} r \varphi(x,v,r) B^\ve(x) \eqlaw
  \rho \iiint r \varphi(x,v,r) B(x) \; dx\;
  d\mu(v,r).
\end{equation}
\end{lemma}

\begin{proof}
   Since $\varphi$ is a smooth function we can approximate both sides by
  step functions, so that it is enough to prove that
  \begin{equation}
    \label{eq:2}
     \lim_{\ve\to 0} \ve \sum_{(x,v,r)\in X^\ve} 1_{[0\le x<1]} B^\ve(x) \eqlaw
  \rho \int_0^1 B(x) \; dx.
\end{equation}
Since in this limit $(v,r)$ are not involved, to simplify notation we will ignore them.
The positive random measure $M^\ve $ on $[0,1]$ defined by
\begin{equation}
  \label{eq:20}
  M^\ve (\psi) =  \ve \sum_{(x,v,r)\in X^\ve} 1_{[0\le x<1]} \psi(x) , \qquad \psi\in \mathcal C(0,1),
\end{equation}
converges a.s. to the Lebesgue measure $dx$ on $[0,1]$.
Then, by the generalization of Slutsky's theorem in Theorem 2.7 (v) in \cite{vandervaart}, 
the couple $(M^\ve, B^\ve)$ converges in law to $(dx, B)$.

Let $F(\mu, \psi)$ be a continuous function on $\mathcal M_+([0,1])\times \mathcal C(0,1)$.
Then $F(M^\ve,  B^\ve) \longrightarrow F(dx, B)$ in law. Apply this to
\begin{equation}
  \label{eq:21}
  F(\mu, \psi) = \int_0^1 \psi(x) d\mu(x),
\end{equation}
and \eqref{eq:2} follows.
\end{proof}
\bibliographystyle{imsart-number} 
\bibliography{hardrods}       


\end{document}